\documentclass{llncs}
\usepackage{makeidx}
\usepackage[latin1]{inputenc}
\usepackage[T1]{fontenc}
\usepackage{amsmath, amssymb}
\usepackage{enumerate}
\usepackage[colorlinks=true, allcolors=blue]{hyperref}
\usepackage{color}
\usepackage{xspace}
\usepackage{comment}
\usepackage{booktabs}
\usepackage{graphicx}
\usepackage{adjustbox}
\usepackage{subcaption}

\numberwithin{equation}{section}
\numberwithin{theorem}{section}
\newtheorem{mylemma}[theorem]{Lemma}

\definecolor{defcolor}{rgb}{0,0,1}

\newcommand{\cond}{\, | \,}

\newcommand{\E}{\mathbb{E}}
\newcommand{\pr}{\mathbb{P}}
\newcommand{\Var}{\operatorname{Var}}

\newcommand{\weq}{\ = \ }

\newcommand{\wle}{\ \le \ }

\newcommand{\wsim}{\ \sim \ }

\newcommand{\cA}{\mathcal{A}}
\newcommand{\cB}{\mathcal{B}}
\newcommand{\cC}{\mathcal{C}}
\newcommand{\cD}{\mathcal{D}}

\newcommand{\cS}{\mathcal{S}}

\newcommand{\cV}{\mathcal{V}}

\newcommand{\Pow}{\mathrm{Pow}}
\newcommand{\mcf}{\mathrm{MCF}}
\newcommand{\Gsub}{G^{(n_0)}}
\newcommand{\tre}{t}
\newcommand{\trm}{\tau}
\newcommand{\degmax}{d_{\rm max}}
\newcommand{\Vz}{V^{(n_0)}}

\newcommand{\TableThreeCycle}{
\begin{tabular}{lcc}
\toprule
3-cycle & $|\cC|$ & $||\cC||$ \\
\midrule
123 & 1 & 3 \\
12, 13, 23 & 3 & 6 \\
\bottomrule
\end{tabular}
}

\newcommand{\TableThreePath}{
\begin{tabular}{lcc}
\toprule
3-path & $|\cC|$ & $||\cC||$ \\
\midrule
1234 & 1 & 4 \\
123, 34 & 2 & 5 \\
234, 12 & 2 & 5 \\
12, 23, 34 & 3 & 6 \\
\bottomrule
\end{tabular}
}

\newcommand{\TableThreePan}{
\begin{tabular}{lcc}
\toprule
3-pan & $|\cC|$ & $||\cC||$ \\
\midrule
1234 & 1 & 4 \\
123, 34 & 2 & 5 \\
134, 12, 23 & 3 & 7 \\
234, 12, 13 & 3 & 7 \\
12, 13, 23, 34 & 4 & 8\\
\bottomrule
\end{tabular}
}

\newcommand{\TableFourPath}{
\begin{tabular}{lcc}
\toprule
4-path & $|\cC|$ & $||\cC||$ \\
\midrule
12345       & 1 & 5 \\
1234, 45   & 2 & 6 \\
2345, 12   & 2 & 6 \\
123, 345         & 2 & 6 \\
1245, 234       & 2 & 7 \\
123, 34, 45     & 3 & 7 \\
234, 12, 45     & 3 & 7 \\
345, 12, 23     & 3 & 7 \\
1245, 23, 34   & 3 & 8 \\
12, 23, 34, 45 & 4 & 8 \\
\bottomrule
\end{tabular}
}

\newcommand{\TableFourCycle}{
\begin{tabular}{lcc}
\toprule
4-cycle & $|\cC|$ & $||\cC||$ \\
\midrule
1234 & 1 & 4 \\
123, 134 & 2 & 6 \\
124, 234 & 2 & 6 \\
123, 14, 34 & 3 & 7 \\
124, 23, 34 & 3 & 7 \\
134, 12, 23 & 3 & 7 \\
234, 12, 14 & 3 & 7 \\
12, 14, 23, 34 & 4 & 8 \\
\bottomrule
\end{tabular}
}

\newcommand{\TableDiamond}{
\begin{tabular}{lcc}
\toprule
Diamond & $|\cC|$ & $||\cC||$ \\
\midrule
1234 & 1 & 4 \\
123, 234 & 2 & 6 \\
123, 24, 34 & 3 & 7 \\
234, 12, 13 & 3 & 7 \\
124, 134, 23 & 3 & 8 \\
124, 13, 23, 34 & 4 & 9 \\
134, 12, 23, 24 & 4 & 9 \\
12, 13, 23, 24, 34 & 5 & 10 \\
\bottomrule
\end{tabular}
}

\newcommand{\TableChair}{
\begin{tabular}{lcc}
\toprule
Chair & $|\cC|$ & $||\cC||$ \\
\midrule
12345     & 1 & 5 \\
1234, 45 & 2 & 6 \\
1345, 23 & 2 & 6 \\
2345, 13 & 2 & 6 \\
123, 345 & 2 & 6 \\
123, 34, 45 & 3 & 7 \\
134, 23, 45 & 3 & 7 \\
234, 13, 45 & 3 & 7 \\
345, 13, 23 & 3 & 7 \\
13, 23, 34, 45 & 4 & 8 \\
\bottomrule
\end{tabular}
}

\newcommand{\TableButterfly}{
\begin{tabular}{lcc}
\toprule
Butterfly & $|\cC|$ & $||\cC||$ \\
\midrule
12345 & 1 & 5 \\
123, 345 & 2 & 6 \\
1234, 35, 45 & 3 & 8 \\
1235, 34, 45 & 3 & 8 \\
1345, 12, 23 & 3 & 8 \\
2345, 12, 13 & 3 & 8 \\
1245, 134, 235 & 3 & 10 \\
1245, 135, 234 & 3 & 10 \\
123, 34, 35, 45 & 4 & 9 \\
345, 12, 13, 23 & 4 & 9 \\
1245,  134, 23, 25 & 4 & 11 \\
1245,  235 ,13, 34 & 4 & 11 \\
1245,  135, 23, 34 & 4 & 11 \\
1245,  234, 13, 35 & 4 & 11 \\
134, 12, 23, 35, 45 & 5 & 11 \\
135, 12, 23, 34, 45 & 5 & 11 \\
234, 12, 13, 35, 45 & 5 & 11 \\
235, 12, 13, 34, 45 & 5 & 11 \\
1245, 13, 23, 34, 35 & 5 & 12 \\
12, 13, 23, 34, 35, 45 & 6 & 12 \\
\bottomrule
\end{tabular}
}

\begin{document}

\title{Moment-based parameter estimation in binomial random intersection graph models}

\mainmatter
\title{Moment-based parameter estimation in binomial random intersection graph models}
\titlerunning{Parameter estimation in random intersection graphs}
\author{Joona Karjalainen \and Lasse Leskel\"a}
\authorrunning{J. Karjalainen and L. Leskel\"a}
\tocauthor{Joona Karjalainen, Lasse Leskel\"a}
\institute{
Aalto University, Espoo, Finland  \\
 \href{http://math.aalto.fi/en/people/joona.karjalainen}{\nolinkurl{math.aalto.fi/en/people/joona.karjalainen}} \\
 \href{http://math.aalto.fi/\~lleskela/}{\nolinkurl{math.aalto.fi/\~lleskela/}}
}
\maketitle

\begin{abstract}
Binomial random intersection graphs can be used as parsimonious statistical models of large and sparse networks, with one parameter for the average degree and another for transitivity, the tendency of neighbours of a node to be connected. This paper discusses the estimation of these parameters from a single observed instance of the graph, using moment estimators based on observed degrees and frequencies of 2-stars and triangles. The observed data set is assumed to be a subgraph induced by a set of $n_0$ nodes sampled from the full set of $n$ nodes. We prove the consistency of the proposed estimators by showing that the relative estimation error is small with high probability for $n_0 \gg n^{2/3} \gg 1$. As a byproduct, our analysis confirms that the empirical transitivity coefficient of the graph is with high probability close to the theoretical clustering coefficient of the model.
\end{abstract}

\keywords{statistical network model, network motif, model fitting, moment estimator, sparse graph, two-mode network, overlapping communities}

\section{Introduction}

Random intersection graphs are statistical network models with overlapping communities. In general, an intersection graph on a set of $n$ nodes is defined by assigning each node $i$ a set of attributes $V_i$, and then connecting those node pairs $\{i,j\}$ for which the intersection $V_i \cap V_j$ is nonempty. When the assignment of attributes is random we obtain a random undirected graph. By construction, this graph has a natural tendency to contain strongly connected communities because any set of nodes $W_k = \{i: V_i \ni k\}$ affiliated with attribute $k$ forms a clique.

The simplest nontrivial model is the binomial random intersection graph $G=G(n,m,p)$ introduced in \cite{Karonski_Scheinerman_Singer-Cohen_1999}, having $n$ nodes and $m$ attributes, where any particular attribute $k$ is assigned to a node $i$ with probability $p$, independently of other node--attribute pairs. A statistical model of a large and sparse network with nontrivial clustering properties is obtained when $n$ is large, $m \sim \beta n$ and $p \sim \gamma n^{-1}$ for some constants $\beta$ and $\gamma$. In this case the limiting model can be parameterised by its mean degree $\lambda = \beta \gamma^2$ and attribute intensity $\mu = \beta \gamma$. By extending the model by introducing random node weights, we obtain a statistical network model which is rich enough to admit heavy tails and nontrivial clustering properties \cite{Bloznelis_2013,Bloznelis_Kurauskas_2015,Deijfen_Kets_2009,Godehardt_Jaworski_2001}. Such models can also be generalised to the directed case \cite{Bloznelis_Leskela_2016_long}. An important feature of this class of models is the analytical tractability related to component sizes \cite{Bloznelis_2010_Largest,Lageras_Lindholm_2008} and percolation dynamics \cite{Ball_Sirl_Trapman_2014,Britton_Deijfen_Lageras_Lindholm_2008}.

In this paper we discuss the estimation of the model parameters based on a single observed instance of a subgraph induced by a set of $n_0$ nodes. We introduce moment estimators for $\lambda$ and $\mu$ based on observed frequencies of 2-stars and triangles, and describe how these can be computed in time proportional to the product of the maximum degree and the number of observed nodes.  We also prove that the statistical network model under study has a nontrivial empirical transitivity coefficient which can be approximated by a simple parametric formula in terms of $\mu$. 

The majority of classical literature on the statistical estimation of network models concerns exponential random graph models \cite{Wasserman_Faust_1994}, whereas most of  the recent works are focused on stochastic block models \cite{Bickel_Chen_Levina_2011} and stochastic Kronecker graphs \cite{Gleich_Owen_2012}. For binomial random intersection graphs with $m \ll n$, it has been shown \cite{Nikoletseas_Raptopoulos_Spirakis_2012} that the underlying attribute assignment can in principle be learned using maximum likelihood estimation.  To the best of our knowledge, the current paper appears to be the first of its kind to discuss parameter estimation in random intersection graphs where $m$ is of the same order as $n$.

The rest of the paper is organised as follows. In Section~\ref{sec:Model} we describe the model and its key assumptions. Section~\ref{sec:Results} summarises the main results. Section~\ref{sec:Experiments} describes numerical simulation experiments for the performance of the estimators. The proofs of the main results are given in Section~\ref{sec:Proofs}, and Section~\ref{sec:Conlusions} concludes the paper.

\section{Model description}
\label{sec:Model}

\subsection{Binomial random intersection graph}
The object of study is an undirected random graph $G=G(n,m,p)$
on node set $\{1,2,\dots,n\}$ with adjacency matrix having diagonal entries $A(i,i)=0$ and off-diagonal entries 
\[
 A(i,j) \weq \min \left( \sum_{k=1}^m B(i,k) B(j,k), \ 1 \right),
\]
where $B(i,k)$ are independent $\{0,1\}$-valued random integers with mean $p$, indexed by $i=1,\dots,n$ and $k=1,\dots,m$. The matrix $B$ represents a random assignment of $m$ attributes to $n$ nodes, both labeled using positive integers, so that $B(i,k)=1$ when attribute $k$ is assigned to node $i$. The set of attributes assigned to node $i$ is denoted by $V_i = \{k: 
B(i,k)=1\}$. Then a node pair $\{i,j\}$ is connected in $G$ if and only if the intersection $V_i \cap V_j$ is nonempty.

\subsection{Sparse and balanced parameter regimes}
We obtain a large and sparse random graph model by considering a sequence of graphs $G(n,m,p)$ with parameters $(n,m,p) = (n_\nu, m_\nu, p_\nu)$ indexed by a scale parameter $\nu \in \{1,2,\dots\}$ such\footnote{For number sequences $f=f_\nu$ and $g=g_\nu$ indexed by integers $\nu \ge 1$, we denote $f \sim g$ if $f_\nu/g_\nu \to 1$ and $f \ll g$ if $f_\nu/g_\nu \to 0$ as $\nu \to \infty$. The scale parameter is usually omitted.} that $n \gg 1$ and $p \ll m^{-1/2}$ as $\nu \to \infty$. In this case a pair of nodes $\{i,j\}$ is connected with probability
\[
 \pr( ij \in E(G) )
 \weq 1 - (1-p^2)^m
 \wsim m p^2,
\]
and the expected degree of a node $i$ is given by
\begin{equation}
 \label{eq:MeanDegree}
 \E \deg_G(i)
 \weq (n-1) \pr( ij \in E(G) )
 \wsim n m p^2.
\end{equation}
Especially, we obtain a large random graph with a finite limiting mean degree $\lambda \in (0,\infty)$ when we assume that
\begin{equation}
 \label{eq:Sparse}
 n \gg 1,
 \qquad
 mp^2 \sim \lambda n^{-1}.
\end{equation}
This will be called the \emph{sparse parameter regime} with mean degree $\lambda$.

The most interesting model with nontrivial clustering properties is obtained when we also assume that $p \sim \mu m^{-1}$
for some constant $\mu \in (0,\infty)$. In this case the full set of conditions is equivalent to
\begin{equation}
 \label{eq:Balanced}
 n \gg 1,
 \qquad
 m \sim (\mu^2/\lambda) n,
 \qquad
 p \sim (\lambda/\mu) n^{-1},
\end{equation}
and will be called as \emph{balanced sparse parameter regime} with mean degree $\lambda$ and attribute intensity $\mu$.

\subsection{Induced subgraph sampling}
\label{sec:Sampling}
Assume that we have observed the subgraph $\Gsub$ of $G$ induced by a set $\Vz$ of $n_0$ nodes sampled from the full set of $n$ nodes, so that $E(\Gsub)$ consists of node pairs $\{i,j\} \in E(G)$ such that $i \in \Vz$ and $j \in \Vz$.
The sampling mechanism used to generate $\Vz$ is assumed to be stochastically independent of $G$. Especially, any nonrandom selection of $\Vz$ fits this framework. On the other hand, several other natural sampling mechanisms \cite{Kolaczyk_2009} are ruled out by this assumption, although we believe that several of the results in this paper can be generalised to a wider context.

In what follows, we shall assume that the size of observed subgraph satisfies $n^\alpha \ll n_0 \le n$ for some $\alpha \in (0,1)$. An important special case with $n_0 = n$ amounts to observing the full graph $G$.

\section{Main results}
\label{sec:Results}

\subsection{Estimation of mean degree}

Consider a random intersection graph $G=G(n,m,p)$ in a sparse parameter regime \eqref{eq:Sparse} with mean degree $\lambda \in (0,\infty)$, and assume that we have observed a subgraph $\Gsub$ of $G$ induced by a set of nodes $\Vz$ of size $n_0$, as described in Section~\ref{sec:Sampling}. Then a natural estimator of $\lambda$ is the normalised average degree
\begin{equation}
 \label{eq:LambdaEstimator}
 \hat \lambda( \Gsub )
 \weq \frac{n}{n_0^2} \sum_{i \in \Vz} \deg_{\Gsub}(i).
\end{equation}
This estimator is asymptotically unbiased because by \eqref{eq:MeanDegree},
\[
 \E \hat \lambda( \Gsub )
 \weq \frac{n}{n_0} (n_0-1) \pr( ij \in E(G) )
 \wsim \lambda.
\]
The following result provides a sufficient condition for the consistency of the estimator of the mean degree $\lambda$, i.e., $\hat \lambda \rightarrow \lambda$ in probability as $n \rightarrow \infty$.

\begin{theorem}
\label{the:EstimatorLambda}
For a random intersection graph $G=G(n,m,p)$ in a sparse parameter regime \eqref{eq:Sparse}, the estimator of $\lambda$ defined by \eqref{eq:LambdaEstimator} is consistent when $n_0 \gg n^{1/2}$. Moreover,
$
  \hat \lambda( \Gsub )
  = \lambda +O_p(n^{1/2}/n_0)
$
for $m \gg n_0^2/n \gg 1$.
\end{theorem}

\subsection{Transitivity coefficient}

For a random or nonrandom graph $G$ with maximum degree at least two, the transitivity coefficient (a.k.a.\ global clustering coefficient \cite{VanDerHofstad_2017,Newman_2003_Structure}) is defined by
\begin{equation}
 \label{eq:tre}
 \tre(G) \weq 3 \frac{N_{K_3}(G)}{N_{S_2}(G)}
\end{equation}
and the \emph{model transitivity coefficient} by
\[
 \trm(G)
 \weq 3 \frac{\E N_{K_3}(G)}{\E N_{S_2}(G)},
\]
where $N_{K_3}(G)$ is the number of triangles\footnote{subgraphs isomorphic to the graph $K_3$ with $V(K_3) = \{1,2,3\}$ and $E(K_3) = \{12,13,23\}$.} and $N_{S_2}(G)$ is the number of 2-stars\footnote{subgraphs isomorphic to the graph $S_2$ with $V(S_2) = \{1,2,3\}$ and $E(S_2) = \{12,13\}$.} in $G$. The above definitions are motivated by noting that
\begin{alignat*}{2}
 \tre(G) &\weq \pr_G( \,  &I_2 I_3 \in E(G) \, \cond \, I_1 I_2 \in E(G), \, I_1 I_3 \in E(G) \, ), \\
 \trm(G) &\weq \, \pr( &I_2 I_3 \in E(G) \, \cond \, I_1 I_2 \in E(G), \, I_1 I_3 \in E(G) \, ),
\end{alignat*}
for an ordered 3-tuple of distinct nodes $(I_1,I_2,I_3)$ selected uniformly at random and independently of $G$, where $\pr_G$ refers to conditional probability given an observed realisation of $G$. The model transitivity coefficient $\trm(G)$ is a nonrandom quantity which depends on the random graph model $G$ only via its probability distribution, and is often easier to analyse than its empirical counterpart. Although $\trm(G) \ne \E \tre(G)$ in general, it is widely believed that $\trm(G)$ is a good approximation of $\tre(G)$ in large and sparse graphs \cite{Bloznelis_2013,Deijfen_Kets_2009}. The following result confirms this in the context of binomial random intersection graphs.

\begin{theorem}
\label{the:Transitivity}
Consider a random intersection graph $G=G(n,m,p)$ in a balanced sparse parameter regime \eqref{eq:Balanced}. If $n_0 \gg n^{2/3}$, then 
\begin{equation}
 \label{eq:Transitivity}
 \tre(\Gsub) \weq \frac{1}{1+\mu} + o_p(1).
\end{equation}
\end{theorem}

It has been observed (with a slightly different parameterisation) in \cite{Deijfen_Kets_2009} that the model transitivity coefficient of the random intersection graph $G=G(n,m,p)$ satisfies
\[
 \trm(G) \weq
 \begin{cases}
 1 + o(1), &\quad p \ll m^{-1}, \\
 \frac{1}{1+\mu} + o(1), &\quad p \sim \mu m^{-1}, \\
 o(1), &\quad m^{-1} \ll p \ll m^{-1/2},
 \end{cases}
\]
and only depends on $n$ via the scale parameter. Hence, as a consequence of Theorem~\ref{the:Transitivity}, it follows that
\[
  \tre(G) \weq  \trm(G) + o_p(1)
\]
for large random intersection graphs $G=G(n,m,p)$ in the balanced sparse parameter regime \eqref{eq:Balanced}.

\subsection{Estimation of attribute intensity}

Consider a random intersection graph $G=G(n,m,p)$ in a balanced sparse parameter regime \eqref{eq:Balanced} with mean degree $\lambda \in (0,\infty)$ and attribute intensity $\mu \in (0,\infty)$, and assume that we have observed a subgraph $\Gsub$ of $G$ induced by a set of nodes $\Vz$ of size $n_0$, as described in Section~\ref{sec:Sampling}. We will now introduce two estimators for the attribute intensity $\mu$.

The first estimator of $\mu$ is motivated by the connection between the empirical and model transitivity coefficients established in Theorem~\ref{the:Transitivity}. By ignoring the error term in \eqref{eq:Transitivity}, plugging the observed subgraph $\Gsub$ into the definition of the transitivity coefficient \eqref{eq:tre}, and solving for $\mu$, we obtain an estimator
\begin{equation}
 \label{eq:MuEstimator1}
 \hat \mu_1( G^{(n_0)} )
 \weq \frac{N_{S_2}(\Gsub)}{3 N_{K_3}(\Gsub)} \, - \, 1.
\end{equation}

An alternative estimator of $\mu$ is given by
\begin{equation}
 \label{eq:MuEstimator2}
 \hat \mu_2( G^{(n_0)} )
 \weq \left( \frac{n_0 N_{S_2}(\Gsub)}{2 N_{K_2}(\Gsub)^2} \, - \, 1 \right)^{-1},
\end{equation}
where $N_{K_2}(\Gsub) = |E(\Gsub)|$. A heuristic derivation of the above formula is as follows. For a random intersection graph $G$ in the balanced sparse parameter regime \eqref{eq:Balanced}, the expected number of 2-stars in $\Gsub$ is asymptotically (see Section \ref{sec:Proofs})
\[
  \E N_{S_2}(\Gsub)
  \wsim 3 \binom{n_0}{3} (mp^3 + m^2 p^4)
  \wsim \frac{1}{2} n_0^3 \mu^3(1+\mu) m^{-2}
\]
and the expectation of $N_{K_2}(\Gsub) = |E(\Gsub)|$ is asymptotically
\[
  \E N_{K_2}(\Gsub)
  \wsim \binom{n_0}{2} mp^2
  \wsim \frac{1}{2} n_0^2 \mu^2 m^{-1}.
\]
Hence
\[
 \frac{\E N_{S_2}(\Gsub)}{(\E N_{K_2}(\Gsub))^2}
 \wsim \frac{2}{n_0} (1 + \mu^{-1}),
\]
so by omitting the expectations above and solving for $\mu$ we obtain \eqref{eq:MuEstimator2}.

The following result confirms that both of the above heuristic derivations yield consistent estimators for the attribute intensity when the observed subgraph is large enough.
\begin{theorem}
\label{the:EstimatorMu}
For a random intersection graph $G=G(n,m,p)$ in a balanced sparse parameter regime \eqref{eq:Balanced}, the estimators of $\mu$ defined by \eqref{eq:MuEstimator1} and \eqref{eq:MuEstimator2} are consistent when $n_0 \gg n^{2/3}$.
\end{theorem}

\subsection{Computational complexity of the estimators}

The evaluation of the estimator $\hat \lambda$ given by \eqref{eq:LambdaEstimator} requires computing the degrees of the nodes in the observed subgraph $\Gsub$. This can be done in $O(n_0 \degmax)$ time, where $\degmax$ denotes the maximum degree of $\Gsub$.

Evaluating the estimator $\hat\mu_1$ given by \eqref{eq:MuEstimator1} requires counting the number of triangles in $\Gsub$ which is a nontrivial task for very large graphs. A naive algorithm requires an overwhelming $O(n_0^3)$ time for this, a listing method can accomplish this in $O(n_0 d_{\rm max}^2)$ time, and there also exist various more advanced algorithms \cite{Tsourakakis_2008}.

The estimator $\hat \mu_2$ given by \eqref{eq:MuEstimator2} can be computed without the need to compute the number of triangles. Actually, the computation of $\hat \mu_2$ only requires to evaluate the degrees of the nodes in $\Gsub$. Namely, with help of the formulas
\[
 N_{K_2}(\Gsub) \weq \frac{1}{2} \sum_{i \in \Vz} \deg_{\Gsub}(i)
 \quad \text{and} \quad
 N_{S_2}(\Gsub) \weq \sum_{i \in \Vz} \binom{\deg_{\Gsub}(i)}{2},
\]
one can verify that
\[
 \hat \mu_2( G^{(n_0)} )
 \weq \left( \frac{a_2 - a_1}{a_1^2} \, - \, 1 \right)^{-1},
\]
where $a_k = n_0^{-1} \sum_{i \in \Vz} \deg_{\Gsub}(i)^k$ denotes the $k$-th moment of the empirical degree distribution of $\Gsub$.

We conclude that the parameters $(\lambda,\mu)$ of the random intersection graph $G=G(n,m,p)$ in the balanced sparse parameter regime \eqref{eq:Balanced} can be consistently estimated in $O(n_0 \degmax)$ time using the estimators $\hat \lambda$ and $\hat \mu_2$.

\section{Numerical experiments}
\label{sec:Experiments}

In this section we study the non-asymptotic behaviour of the parameter estimators $\hat \lambda$ \eqref{eq:LambdaEstimator}, $\hat \mu_1$ \eqref{eq:MuEstimator1}, and $\hat \mu_2$ \eqref{eq:MuEstimator2} using simulated data. In the first experiment, a random intersection graph was generated for each $n=50, 70, \ldots , 1000$, using parameter values $(\lambda = 9,$ $ \mu=3)$ and $(\lambda = 2, $ $\mu = 0.5)$. All of the data was used for estimation, i.e., $n_0 = n$.

\begin{figure}[h]
      \begin{subfigure}{0.5\textwidth}
        \includegraphics[trim = 20mm 67mm 15mm 70mm, clip,width=\textwidth]{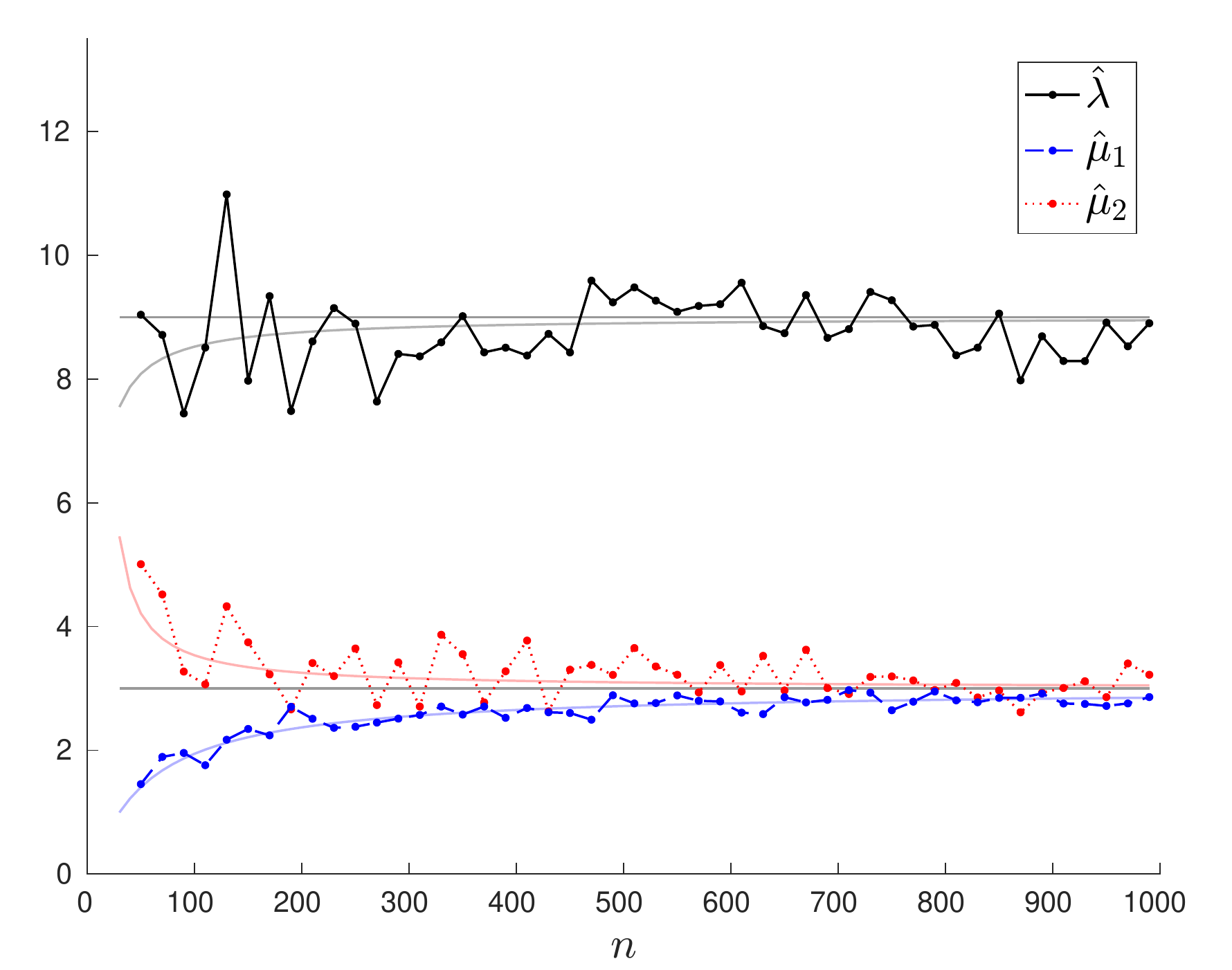}
         \caption{ $\lambda = 9$, $\mu = 3$}
      \end{subfigure}
      \begin{subfigure}{0.5\textwidth}
        \includegraphics[trim = 20mm 68.5mm 15mm 68.5mm, clip,width=\textwidth]{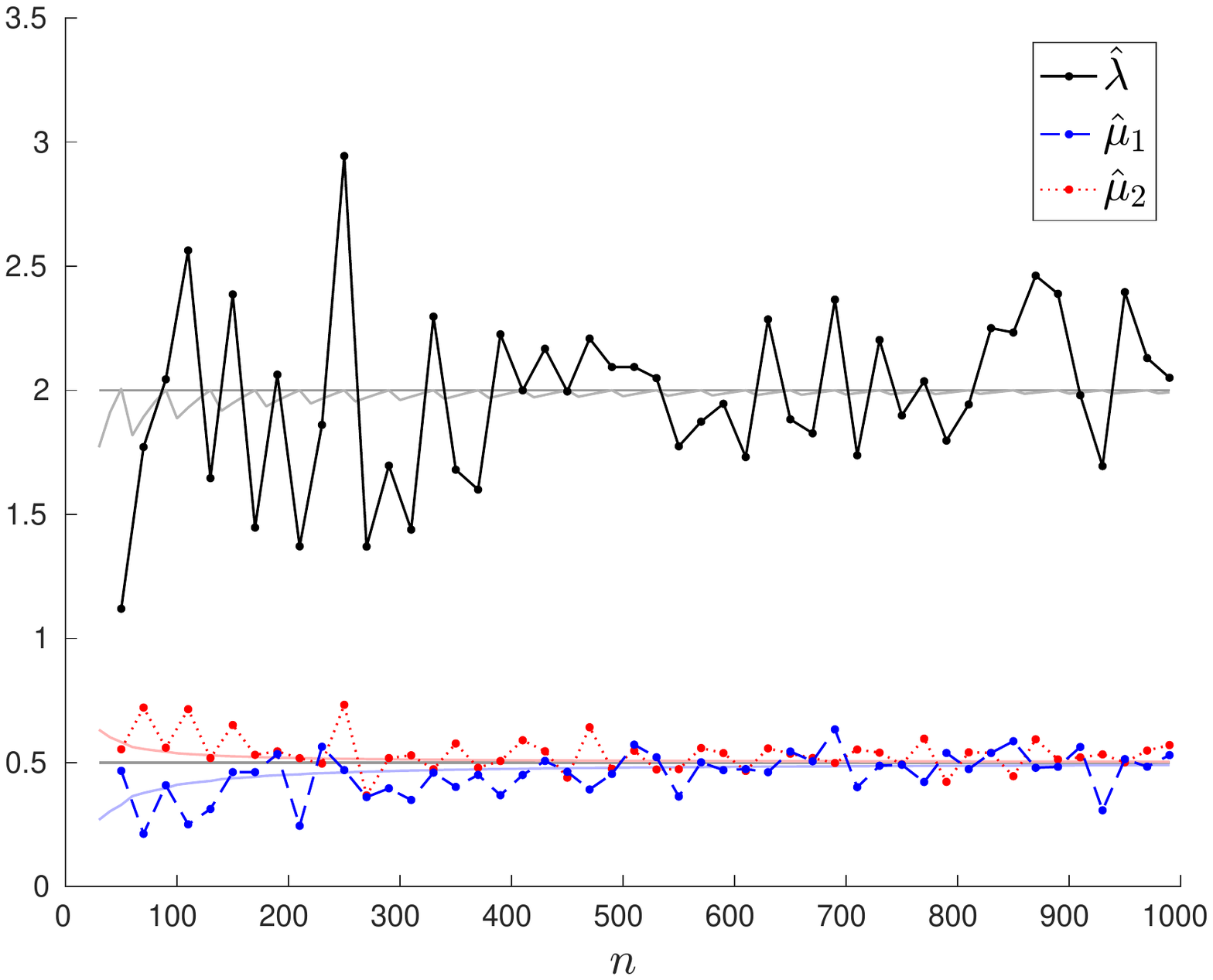}
         \caption{ $\lambda = 2$, $\mu = 0.5$} 
      \end{subfigure}

\caption{  Simulated values of the estimators $\hat \lambda$, $\hat \mu_1$, and $\hat \mu_2$ with $n_0 = n$. The solid curves show the theoretical values of the estimators when the feature counts $N_*(G^{(n)})$ are replaced by their expected values.}
\label{fig:plot1000}
\end{figure}

\begin{figure}[h]
\begin{subfigure}{0.5\textwidth}
        \includegraphics[trim = 20mm 67mm 10mm 70mm, width=\textwidth]{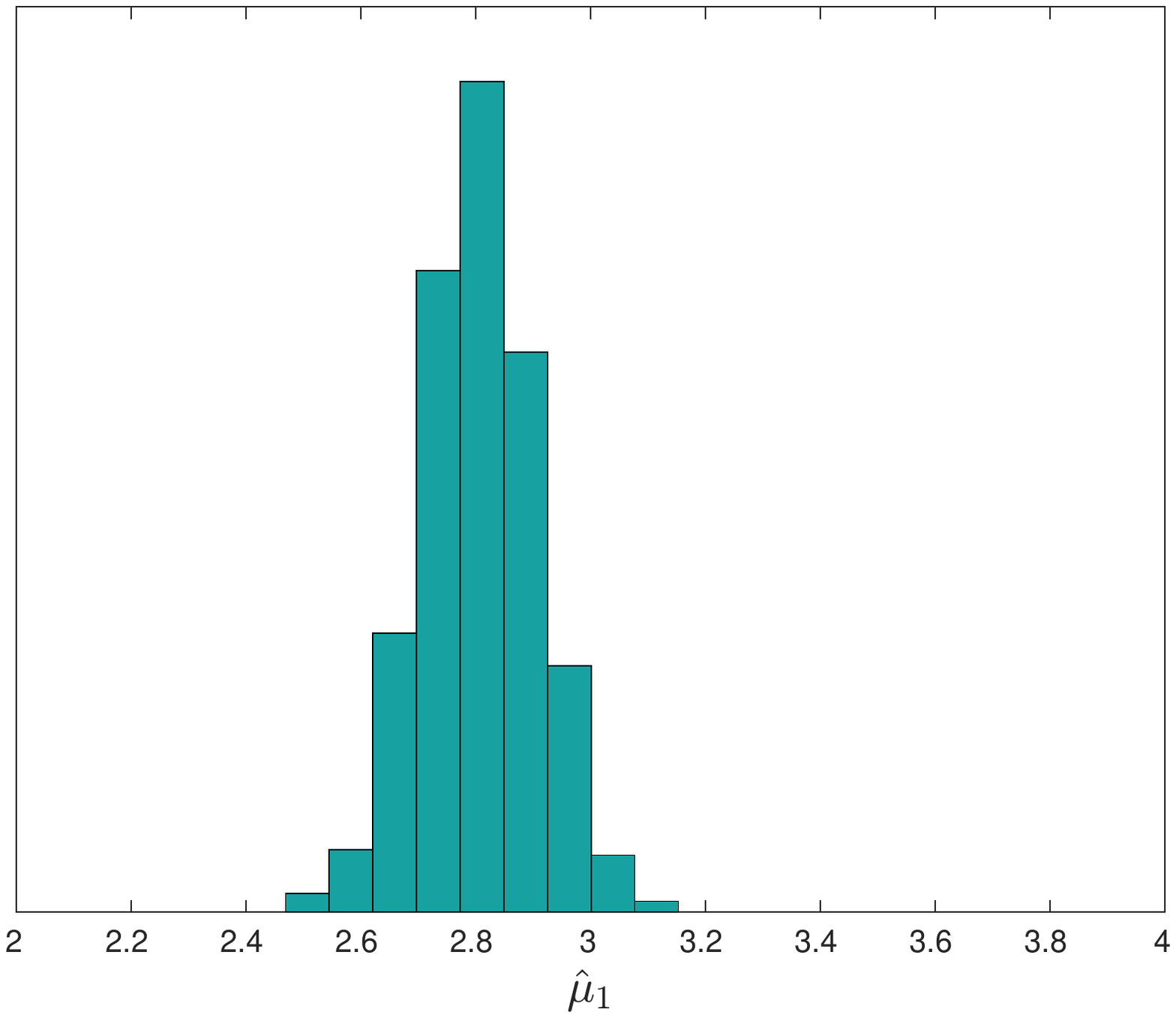} 
      \end{subfigure}
      \begin{subfigure}{0.5\textwidth}
        \includegraphics[trim = 20mm 65mm 2mm 68mm, width=\textwidth]{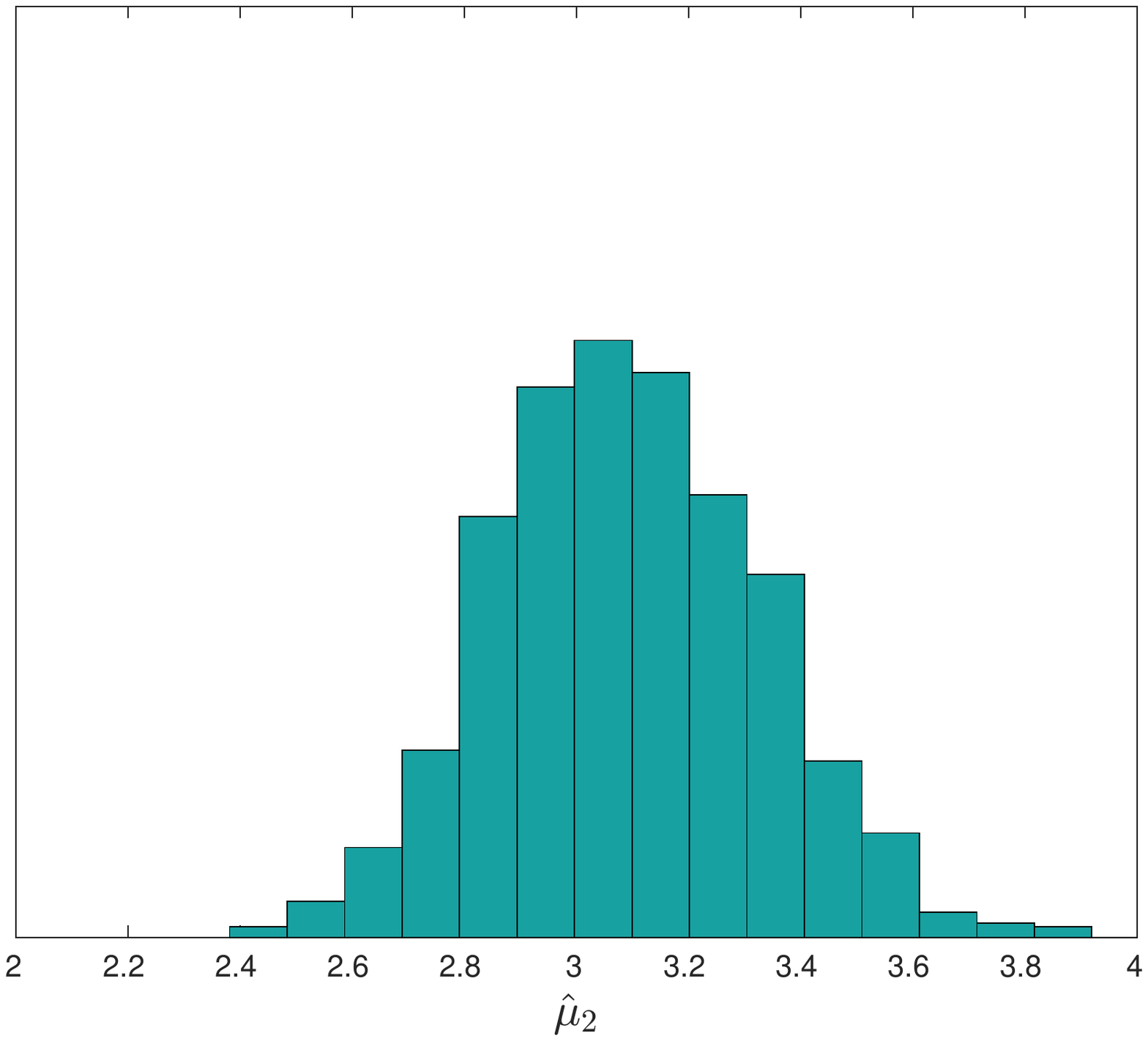}
      \end{subfigure}
      \caption{1000 simulated values of $\hat \mu_1$ and $\hat \mu_2$ with $\lambda = 9,$ $ \mu = 3$, and $n_0 = n = 750$.}
\label{fig:hist750}
\end{figure}

Figure \ref{fig:plot1000} shows the computed estimates $\hat \lambda$, $\hat \mu_1$, and $\hat \mu_2$ for each $n$. For comparison, the theoretical values of these estimators are also shown when the counts of links, 2-stars, and triangles are replaced by their expected values in \eqref{eq:LambdaEstimator}, \eqref{eq:MuEstimator1}, and \eqref{eq:MuEstimator2}. 

With $(\lambda=9,$ $ \mu = 3)$, the parameter $\mu$ is generally underestimated by $\hat \mu_1$ and overestimated by $\hat \mu_2$. The errors in $\hat \mu_1$ appear to be dominated by the bias, whereas the errors in $\hat \mu_2$ are mostly due to variance.
With $(\lambda=2,$ $ \mu = 0.5)$, the simulated graphs are more sparse. The differences between the two estimators of $\mu$ are small, and the relative error of $\hat \lambda$ appears to have increased. The discontinuities of the theoretical values of $\hat \lambda$ are due to the rounding of the numbers of attributes $m$.

In the second experiment, 1000 random intersection graphs were simulated with $n_0 = n = 750$ and $(\lambda = 9,$ $ \mu = 3)$. Histograms of the estimates of $\mu$ are shown in Figure \ref{fig:hist750}. The bias is visible in both $\hat \mu_1$ and $\hat \mu_2$, and the variance of $\hat \mu_2$ is larger than that of $\hat \mu_1$. However, the difference in accuracy is counterbalanced by the fact that $\hat \mu_1$ requires counting the triangles.


\section{Proofs}
\label{sec:Proofs}

\subsection{Covering densities of subgraphs}
\label{sec:covdens}
Denote by $\Pow(\Omega)$ the collection of all subsets of $\Omega$. For $\cA, \cB \subset \Pow(\Omega)$ we denote $\cA \Subset \cB$ and say that $\cB$ is a \emph{covering family} of $\cA$, if for every $A \in \cA$ there exists $B \in \cB$ such that $A \subset B$. A covering family $\cB$ of $\cA$ is called \emph{minimal} if for any $B \in \cB$,
\begin{enumerate}[(i)]
\item the family obtained by removing $B$ from $\cB$ is not a covering family of $\cA$, and
\item the family obtained by replacing $B$ by a strict subset of $B$ is not a covering family of $\cA$.
\end{enumerate}
For a graph $R = (V(R),E(R))$, we denote by $\mcf(R)$ the set of minimal covering families of $E(R)$. Note that all members of a minimal covering family have size at least two. For a family of subsets $\cC = \{C_1,\dots,C_t\}$ consisting of $t$ distinct sets, we denote $|\cC| = t$ and $||\cC|| = \sum_{s=1}^t |C_s|$. The notation $R \subset G$ means that $R$ is a subgraph of $G$.

The following result is similar in spirit to \cite[Theorem 3]{Karonski_Scheinerman_Singer-Cohen_1999}, but focused on subgraph frequencies instead of appearance thresholds.

\begin{theorem}
\label{the:MotifDensity}
If $mp^2 \ll 1$, then for any finite graph $R$ not depending on the scale parameter,
\[
 \pr(G \supset R)
 \wsim \sum_{\cC \in \mcf(R)} m^{|\cC|} p^{||\cC||}
\]
\end{theorem}

The proof of Theorem~\ref{the:MotifDensity} is based on two auxiliary results which are presented first.

\begin{mylemma}
\label{the:IntersectionGraph}
For any intersection graph $G$ on $\{1,\dots,n\}$ generated by attribute sets $\cV = \{V_1, \dots, V_n\}$ and any graph $R$ with $V(R) \subset V(G)$, the following are equivalent:
\begin{enumerate}[(i)]
\item $R \subset G$.
\item $E(R) \Subset \cV$.
\item There exists a family $\cC \in \mcf(R)$ such that $E(R) \Subset \cC \Subset \cV$.
\end{enumerate}
\end{mylemma}
\begin{proof}
(i)$\iff$(ii). Observe that a node pair $e \in \binom{V}{2}$ satisfies $e \in E(G)$ if and only if $e \subset V_j$ for some $V_j \in \cV$. Hence $E(R) \subset E(G)$ if and only if for every $e \in E(R)$ there exists $V_j \in \cV$ such that $e \subset V_j$, or equivalently, $E(R) \Subset \cV$.

(ii)$\implies$(iii). If $E(R) \Subset \cV$, define $C_j = V_j \cap V(R)$. Then $\cC = \{C_1, \dots, C_m\}$ is a covering family of $E(R)$. Then test whether $\cC$ still remains a covering family of $E(R)$ if one its members is removed. If yes, remove the member of $\cC$ with the highest label. Repeat this procedure until we obtain a covering family $\cC'$ of $E(R)$ for which no member can be removed. Then test whether some $C \in \cC'$ can be replaced by a strict subset of $C$. If yes, do this replacement, and repeat this procedure until we obtain a covering family $\cC''$ of $E(R)$ for which no member can be shrunk in this way. This mechanism implies that $\cC''$ is a minimal covering family of $E(R)$, for which $E(R) \Subset \cC'' \Subset \cV$.

(iii)$\implies$(ii). Follows immediately from the transitivity of $\Subset$.
\end{proof}

\begin{mylemma}
\label{the:Covering}
If $mp^2 \ll 1$, then for any scale-independent finite collection $\cC = \{C_1,\dots, C_t\}$ of finite subsets of $\{1,2,\dots\}$ of size at least 2, the probability that the family of attribute sets $\cV = \{V_1, \dots, V_n\}$ of $G=G(n,m,p)$ is a covering family of $\cC$ satisfies
\[
 \pr( \cV \Supset \cC )
 \wsim m^{|\cC|} p^{||\cC||}.
\]
\end{mylemma}
\begin{proof}
For $s=1,\dots, t$, denote by $N_s = \sum_{j=1}^m 1(V_j \supset C_s)$ the number of attribute sets covering $C_s$. Note that $N_s$ follows a binomial distribution with parameters $m$ and $p^{|C_s|}$. Because $|C_s| \ge 2$, it follows that the mean of $N_s$ satisfies $m p^{|C_s|} \le mp^2 \ll 1$. Using elementary computations related to the binomial distribution (see e.g.\ \cite[Lemmas 1,2]{Karonski_Scheinerman_Singer-Cohen_1999}) it follows that the random integers $N_1,\dots,N_t$ are asymptotically independent with $\pr(N_s \ge 1) \sim m p^{|C_s|}$, so that 
\[
 \pr(\cV \Supset \cC)
 \weq \pr(N_1 \ge 1, \dots, N_t \ge 1)
 \wsim \prod_{s=1}^t \pr(N_s \ge 1) 
 \wsim m^t p^{\sum_{s=1}^t |C_s|}.
\]
\end{proof}

\begin{proof}[Proof of Theorem~\ref{the:MotifDensity}]
By Lemma~\ref{the:IntersectionGraph}, we see that
\[
 \pr(G \supset R)
 \weq \pr \left( \bigcup_{\cC \in \mcf(R)} \{\cV \Supset \cC\} \right).
\]
Bonferroni's inequalities hence imply $U_1 - U_2 \le \pr(G \supset R) \le U_1$, where
\[
 U_ 1
 \weq  \sum_{\cC \in \mcf(R)} \pr(\cV \Supset \cC)
 \qquad\text{and}\qquad
 U_2
 \weq \sum_{\cC, \cD} \pr(\cV \Supset \cC, \cV \Supset \cD),
\]
and the latter sum is taken over all unordered pairs of distinct minimal covering families $\cC, \cD \in \mcf(R)$. Note that by Lemma~\ref{the:Covering},
\[
 U_1 \wsim \sum_{\cC \in \mcf(R)} m^{|\cC|} p^{||\cC||},
\]
so to complete the proof it suffices to verify that $U_2 \ll U_1$.

Fix some minimal covering families $\cC = \{C_1,\dots, C_s\}$ and $\cD = \{D_1,\dots, D_t\}$ of $E(R)$ such that $\cC \ne \cD$. Then either $\cC$ has a member such that $C_i \not\in \cD$, or 
$\cD$ has a member such that $D_j \not\in \cC$. In the former case $\cC \cup \cD \supset \{C_i,D_1, \dots, D_t\}$, so that by Lemma~\ref{the:Covering},
\begin{align*}
 \pr( \cV \Supset \cC \cup \cD )
 \wle \pr( \cV \Supset \{C_i,D_1, \dots, D_t\} )
 &\wsim m^{t+1} p^{|C_i| + \sum_{j=1}^t |D_j|} \\
 &\wsim m p^{|C_i|} \pr(\cV \Supset \cD).
\end{align*}
Because $\cC$ is a minimal covering family, $|C_i| \ge 2$, and $m p^{|C_i|} \le mp^2 \ll 1$, and hence 
$\pr( \cV \Supset \cC \cup \cD ) \ll \pr(\cV \Supset \cD)$.
In the latter case where $\cD$ has a member such that $D_j \not\in \cC$, a similar reasoning shows that $\pr( \cV \Supset \cC \cup \cD ) \ll \pr(\cV \Supset \cC)$. We may hence conclude that
\[
 \pr(\cV \Supset \cC, \cV \Supset \cD)
 \weq \pr( \cV \Supset \cC \cup \cD )
 \ \ll \ \pr(\cV \Supset \cC) + \pr(\cV \Supset \cD)
\]
for all distinct $\cC, \cD \in \mcf(R)$. Therefore, the proof is completed by
\[
 U_2
 \ \ll \ \sum_{\cC,\cD} \Big( \pr(\cV \Supset \cC) + \pr(\cV \Supset \cD) \Big)
 \wle 2 |\mcf(R)| U_1.
\]
\end{proof}

\subsection{Covering densities of certain subgraphs}
In order to bound the variances of subgraph counts we will use the covering densities of (partially) overlapping pairs of 2-stars and triangles. Figure~\ref{fig:Triangles} displays the graphs obtained as a union of two partially overlapping triangles. Figure~\ref{fig:2-stars} displays the graphs produced by overlapping 2-stars.

\begin{figure}[h]
\begin{center}
\includegraphics[width=.5\textwidth]{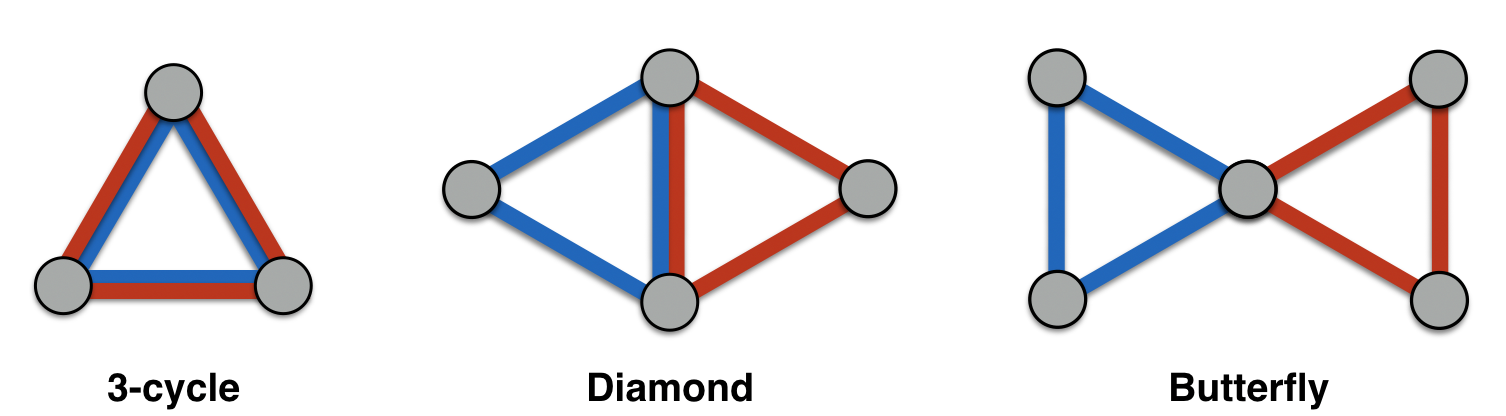}
\caption{\label{fig:Triangles} Graphs obtained as unions of overlapping triangles.}
\end{center}
\end{figure}

\begin{figure}[h]
\begin{center}
\includegraphics[width=.5\textwidth]{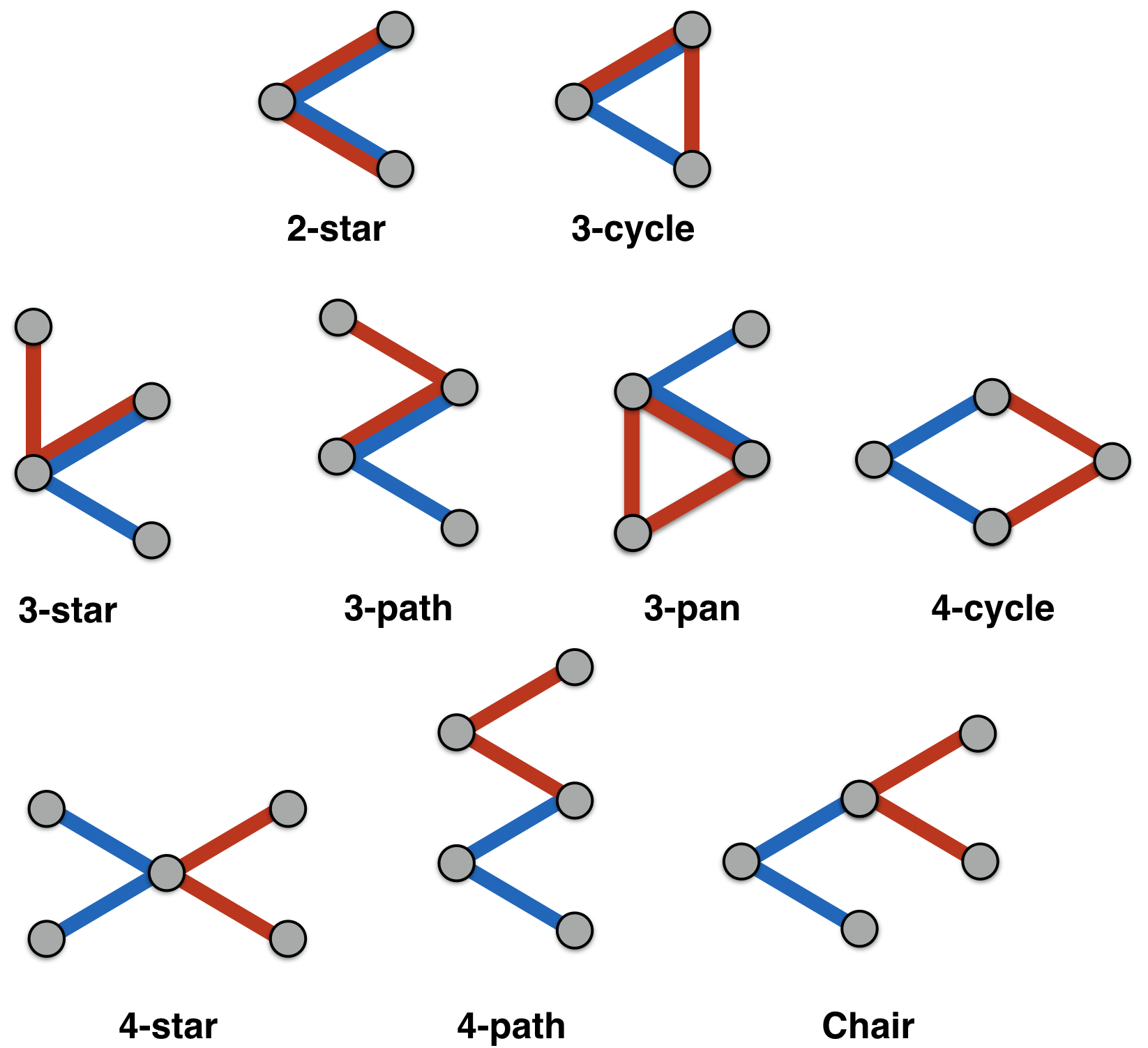}
\caption{\label{fig:2-stars} Graphs obtained as unions of overlapping 2-stars.}
\end{center}
\end{figure}

According to Theorem~\ref{the:MotifDensity}, the covering densities of subgraphs may be computed from their minimal covering families.  For a triangle $R$ with $V(R) = \{1,2,3\}$ and $E(R) = \{12,13,23\}$, the minimal covering families are\footnote{For clarity, we write 12 and 123 as shorthands of the sets $\{1,2\}$ and $\{1,2,3\}$.} $\{123\}$ and $\{12,13,23\}$. The minimal covering families of the a 3-path $R$ with $V(R) = \{1,2,3,4\}$ and $E(R) = \{12,23,34\}$ are given by $\{1234\}$, $\{12,234\}$, $\{123,34\}$, and $\{12,23,34\}$.

The covering densities of stars are found as follows.
Fix $r \ge 1$, and let $R$ be the $r$-star such that $V(R) = \{1,2,\dots,r+1\}$ and $E(R) = \{\{1,r+1\},\{2,r+1\},\dots,\{r,r+1\}\}$. The minimal covering families of $R$ are of the form $\cC = \{ S \cup \{r+1\}:  S \in \cS \}$, where $\cS$ is a partition of the leaf set $\{1,\dots,r\}$ into nonempty subsets. For any such $\cC$ we have $|\cC| = |\cS|$ and $||\cC|| = r + |\cS|$. Hence
\[
 \pr(G \supset \text{$r$-star})
 \wsim \sum_{k=1}^r {r \brace k} m^k p^{k+r},
\]
where ${r \brace k}$ equals the number of partitions of $\{1,\dots,r\}$ into $k$ nonempty sets. These coefficients are known as Stirling numbers of the second kind \cite{Graham_Knuth_Patashnik_1994} and can be computed via ${r \brace k} = \frac{1}{k!} \sum_{j=0}^k (-1)^{k-j} \binom{k}{j}j^r$. Hence,
\[
 \pr(G \supset \text{$r$-star})
 \wsim
 \begin{cases}
   mp^3 + m^2 p^4, &\quad r=2,\\
   mp^4 + 3m^2 p^5 + m^3 p^6, &\quad r=3,\\
   mp^5 + 7m^2 p^6 + 6m^3 p^7 + m^4 p^8, &\quad r=4.
 \end{cases}
\]

Table~\ref{tab:Summary} summarises approximate covering densities of overlapping pairs of 2-stars and triangles. The table is computed by first listing all minimal covering families of the associated subgraphs, as shown in Table~\ref{tab:ManyOtherOther}. We also use the following observations (for $p \ll m^{-1/2} \ll 1$) to cancel some of the redundant terms in the expressions.

4-path: $m^2p^7 \ll m^2 p^6$ and $m^3 p^8 \ll m^3 p^7$

4-cycle: $m^2 p^6 \ll mp^4$ 

3-pan: $m^3 p^7 \ll m^2 p^5$

Diamond: $m^2 p^6 \ll mp^4$ and $m^4 p^9 \ll m^3 p^7$

Butterfly: $m^5 p^{12} \ll m^5 p^{11}$, $m^4 p^{11} \ll m^3 p^8 \ll m^2 p^6$, $m^3 p^{10} \le m^3 p^8 \ll m^2 p^6$.

\begin{table}[h]
\small
\begin{center}
\begin{tabular}{lccll}
\toprule
$R$ & $|V(R)|$ & $|E(R)|$ & Appr.\ density ($p \ll m^{-1/2} \ll 1$) & Appr.\ density ($p \sim \mu m^{-1}$) \\
\midrule
1-star & 2 & 1 & $ mp^2$ & $\mu^2 m^{-1}$ \\[1.5ex]
2-star & 3 & 2 & $ mp^3 + m^2 p^4$ & $(1+\mu)\mu^3 m^{-2}$ \\
3-cycle & 3 & 3 & $ mp^3 + m^3 p^6$ & $\mu^3 m^{-2}$  \\[1.5ex]
3-star    & 4 & 3 & $ mp^4 + 3m^2 p^5 + m^3 p^6$ & $(1+3\mu+\mu^2)\mu^4 m^{-3}$ \\
3-path   & 4 & 3 & $ mp^4 + 2m^2 p^5 + m^3 p^6$ & $(1+2\mu+\mu^2)\mu^4 m^{-3}$ \\
4-cycle  & 4 & 4 & $ mp^4 + 4m^3 p^7 + m^4 p^8$ & $\mu^4 m^{-3}$ \\
3-pan & 4 & 4 & $ mp^4 + m^2 p^5 + m^4 p^8$ & $(1+\mu)\mu^4 m^{-3}$ \\
Diamond & 4 & 5 & $ mp^4 + 2 m^3 p^7 + m^5 p^{10}$ & $\mu^4 m^{-3}$\\[1.5ex]
4-star    & 5 & 4 & $ mp^5 + 7m^2 p^6 + 6 m^3 p^7 + m^4 p^8$ & $(1+7\mu+6\mu^2+\mu^3)\mu^5 m^{-4}$ \\
4-path   & 5 & 4 & $ mp^5 + 3 m^2 p^6 + 3 m^3 p^7 + m^4 p^8$ & $(1+3\mu+3\mu^2+\mu^3)\mu^5 m^{-4}$ \\
Chair & 5 & 4 & $ mp^5 + 4 m^2 p^6 + 4 m^3 p^7 + m^4 p^8$ & $(1+4\mu+4\mu^2+\mu^3)\mu^5 m^{-4}$ \\
Butterfly & 5 & 6 & $ mp^5 + m^2 p^6 + 2 m^4 p^9 + 4 m^5 p^{11} + m^6 p^{12}$ & $(1+\mu)\mu^5 m^{-4}$ \\
\bottomrule
\end{tabular}
\end{center}
\caption{\label{tab:Summary} Approximate densities of some subgraphs.}
\end{table}

\begin{table}[h]
\begin{center}
\scriptsize
\adjustbox{valign=t}{\begin{minipage}{0.32\textwidth}
\TableThreePath\\[4ex]
\TableFourCycle\\[4ex]
\TableDiamond
\end{minipage}}
\adjustbox{valign=t}{\begin{minipage}{0.32\textwidth}
\TableThreeCycle\\[4ex]
\TableChair\\[4ex]
\TableFourPath\\[4ex]
\end{minipage}}
\adjustbox{valign=t}{\begin{minipage}{0.32\textwidth}
\TableThreePan\\[4ex]
\TableButterfly
\end{minipage}}
\end{center}
\caption{\label{tab:ManyOtherOther} Minimal covering families of the subgraphs in Fig.~\ref{fig:Triangles} and Fig.~\ref{fig:2-stars} (stars excluded).}
\end{table}

\subsection{Proofs of Theorems \ref{the:EstimatorLambda}, \ref{the:Transitivity}, and \ref{the:EstimatorMu}}
\begin{proof}[of Theorem \ref{the:EstimatorLambda}]
Denote $\hat\lambda = \hat\lambda(\Gsub)$ and $\hat N = N_{K_2}(\Gsub)$.  Then the variance of $\hat\lambda$
is given by
\begin{equation}
 \label{eq:LambdaVar}
 \Var(\hat\lambda)
 \weq 4 \frac{n^2}{n_0^4} \Var(\hat N).
\end{equation}
By writing
\[
 \hat N = \sum_{e \in \binom{[n_0]}{2}} 1(G \supset e)
 \qquad\text{and}\qquad 
 \hat N^2 = \sum_{e \in \binom{[n_0]}{2}} \sum_{e' \in \binom{[n_0]}{2}} 1(G \supset e) 1(G \supset e'),
\]
we find that
$
 \E \hat N
 = \binom{n_0}{2} \pr(G \supset K_2)
$
and
\[
 \E \hat N^2
 \weq \binom{n_0}{2} \pr(G \supset K_2) + 2(n_0-2) \binom{n_0}{2} \pr(G \supset S_2) + \binom{n_0}{2} \binom{n_0-2}{2} \pr(G \supset K_2)^2.
\]
Because the last term above is bounded by
\[ 
 \binom{n_0}{2} \binom{n_0-2}{2} \pr(G \supset K_2)^2
 \wle \binom{n_0}{2}^2 \pr(G \supset K_2)^2
 \weq (\E \hat N)^2,
\]
it follows that
\begin{align*}
 \Var(\hat N)
 &\wle \binom{n_0}{2} \pr(G \supset K_2) + 2(n_0-2) \binom{n_0}{2} \pr(G \supset S_2) \\
 &\weq (1+o(1)) \frac{1}{2} n_0^2 mp^2 \ + \ (1+o(1)) n_0^3 (mp^3 + m^2 p^4).
\end{align*}
Hence by \eqref{eq:LambdaVar},
\begin{align*}
 \Var(\hat \lambda)
 &\weq O(n_0^{-2} n^2 mp^2) \ + \ O(n_0^{-1} n^2 mp^3) + O(n_0^{-1} n^2 m^2 p^4),
\end{align*}
and by noting that $n^2 mp^2 \sim \lambda n$, $n^2 m p^3 = m^{-1/2} n^{1/2} (nmp^2)^{3/2} \sim \lambda^{3/2} m^{-1/2} n^{1/2}$ and $n^2 mp^4 = (nmp^2)^2 \sim \lambda^2$, we find that
\[
 \Var(\hat \lambda)
 \weq O\left( n_0^{-2} n + m^{-1/2}  n_0^{-1} n^{1/2} + n_0^{-1} \right)
 \weq O\left( n_0^{-2} n + m^{-1/2}  n_0^{-1} n^{1/2} \right),
\]
where the last equality is true because $n_0^{-2} n \ge n_0^{-1}$. The claim now follows by Chebyshev's inequality.
\end{proof}

\begin{proof}[of Theorem \ref{the:Transitivity} and Theorem \ref{the:EstimatorMu}]
The variances of $N_{S_2}$ and $N_{K_3}$ can be bounded from above in the same way that the variance of $N_{K_2}$ was bounded in the proof of Theorem \ref{the:EstimatorLambda}. 
The overlapping subgraphs contributing to the variance of $N_{K_3}$ are those shown in Fig.~\ref{fig:Triangles}. According to Table \ref{tab:Summary}, the contribution of these subgraphs is  $O(n_0^{|V(R)|}m^{-|V(R)|+1})$ for $|V(R)|=3,4,5$, and the nonoverlapping triangles contribute $O(n_0^{6}m^{-5})$. Since $\E N_{K_3}$ is of the order $n_0^3 m^{-2} $, it follows that $\Var (N_{K_3}/\E N_{K_3}) = o(1)$ for $n_0 \gg n^{2/3}$.

The same line of proof works for $N_{S_2}$, i.e., we note that the subgraphs appearing in $\Var(N_{S_2})$ are those shown in Fig.~\ref{fig:2-stars} and their contributions to the variance are listed in Table \ref{tab:Summary}. Again, it follows that $\Var (N_{S_2}/\E N_{S_2}) = o(1)$ for $n_0 \gg n^{2/3}$. Hence we may conclude using Chebyshev's inequality that
\begin{align*}
 N_{K_3}(\Gsub) &\weq (1+o_p(1)) \E N_{K_3}(\Gsub) \weq (1+o_p(1)) \binom{n_0}{3} \mu^3 m^{-2} \\
 N_{S_2}(\Gsub) &\weq (1+o_p(1)) \E N_{S_2}(\Gsub) \weq (1+o_p(1)) 3 \binom{n_0}{3}(1+\mu) \mu^3 m^{-2},
\end{align*}
and the claim of Theorem \ref{the:Transitivity} follows.

Further, in the proof of Theorem \ref{the:EstimatorLambda} we found that
\[
 N_{K_2}(\Gsub) \weq (1+o_p(1)) \E N_{K_2}(\Gsub) \weq (1+o_p(1)) \binom{n_0}{2} \mu^2 m^{-1}.
\]
Hence the claims of Theorem \ref{the:EstimatorMu} follow from the above expressions combined with the continuous mapping theorem.
\end{proof}

\section{Conclusions}
\label{sec:Conlusions}

In this paper we discussed the estimation of parameters for a large random intersection graph model in a balanced sparse parameter regime characterised by mean degree $\lambda$ and attribute intensity $\mu$, based on a single observed instance of a subgraph induced by a set of $n_0$ nodes. We introduced moment estimators for $\lambda$ and $\mu$ based on observed frequencies of 2-stars and triangles, and described how the estimators can be computed in time proportional to the product of the maximum degree and the number of observed nodes.  We also proved that in this parameter regime the statistical network model under study has a nontrivial empirical transitivity coefficient which can be approximated by a simple parametric formula in terms of $\mu$. 

For simplicity, our analysis was restricted to binomial undirected random intersection graph models, and the statistical sampling scheme was restricted induced subgraph sampling, independent of the graph structure. Extension of the obtained results to general directed random intersection graph models with general sampling schemes is left for further study and forms a part of our ongoing work.

{\small
\subsubsection*{Acknowledgments.} Part of this work has been financially supported by the Emil Aaltonen Foundation, Finland. We thank Mindaugas Bloznelis for helpful discussions, and the two anonymous reviewers for helpful comments.
}

\bibliographystyle{splncs03}
\bibliography{lslReferences}
\end{document}